\documentclass{article} 
\usepackage{nips12submit_e, times, amsthm, amsmath, bbm, amsfonts, amssymb, cite, graphicx, caption, enumerate, stmaryrd, ctable, subfigure, hyperref, cancel}

\title{Towards a learning-theoretic analysis of\\ spike-timing dependent plasticity}

\newenvironment{thmtag}[2][Theorem]{\begin{trivlist}
\item[\hskip \labelsep {\bfseries #1}\hskip \labelsep {\bfseries #2}]}{\end{trivlist}}

\newcommand{\eod}{{${}$\\}}
\newcommand{\X}{{\mathbf X}}
\newcommand{\x}{{\mathbf x}}
\newcommand{\Xn}{{\mathbf{A}}}
\newcommand{\Yn}{{B}}
\newcommand{\xn}{{\mathbf{a}}}
\newcommand{\yn}{{b}}
\newcommand{\wt}{\mathbf{w}}
\newcommand{\bE}{{\mathbb E}}
\newcommand{\bV}{{\mathbb V}}
\newcommand{\bO}{{\mathbbm 1}}
\newcommand{\bR}{{\mathbb R}}
\newcommand{\cF}{{\mathcal F}}
\newcommand{\cA}{{\mathcal A}}
\newcommand{\cL}{{\mathcal L}}

\newtheorem{thm}{Theorem}
\newtheorem{cor}[thm]{Corollary}

\newtheorem{lem}[thm]{Lemma}
\newtheorem{rem}{Remark}
\newtheorem{defn}{Definition}

\author{David Balduzzi\\
MPI for Intelligent Systems, T{\"u}bingen, Germany\\
ETH Zurich, Switzerland\\
\texttt{david.balduzzi@inf.ethz.ch}
\And
Michel Besserve\\
MPI for Intelligent Systems and MPI for Biological Cybernetics\\
T{\"u}bingen, Germany\\
\texttt{michel.besserve@tuebingen.mpg.de}}

\nipsfinalcopy 

\begin{document}
\maketitle

\begin{abstract}
This paper suggests a learning-theoretic perspective on how synaptic plasticity benefits global brain functioning. We introduce a model, the selectron, that {\rm (i)} arises as the fast time constant limit of leaky integrate-and-fire neurons equipped with spiking timing dependent plasticity (STDP) and {\rm (ii)} is amenable to theoretical analysis. We show that the selectron encodes reward estimates into spikes and that an error bound on spikes is controlled by a spiking margin and the sum of synaptic weights. Moreover, the efficacy of spikes (their usefulness to other reward maximizing selectrons) also depends on total synaptic strength. Finally, based on our analysis, we propose a regularized version of STDP, and show the regularization improves the robustness of neuronal learning when faced with multiple stimuli.
\end{abstract}

\section{Introduction}

Finding principles underlying learning in neural networks is an important problem for both artificial and biological networks. An elegant suggestion is that global objective functions may be optimized during learning \cite{friston:06}. For biological networks however, the currently known neural plasticity mechanisms use a very restricted set of data -- largely consisting of spikes and diffuse neuromodulatory signals. How a \emph{global} optimization procedure could be implemented at the \emph{neuronal} (cellular) level is thus a difficult problem.

A successful approach to this question has been Rosenblatt's perceptron \cite{Rosenblatt:1958fk} and its extension to multilayer perceptrons via backpropagation \cite{rumelhart:86}. Similarly, (restricted) Boltzmann machines, constructed from simple stochastic units, have provided a remarkably powerful approach to organizing distributed optimization across many layers \cite{hinton:06a}. By contrast, although there has been significant progress in developing and understanding more biologically realistic models of neuronal learning \cite{song:00, seung:03, bohte:05, legenstein:06, buesing:07, legenstein:08a}, these do not match the performance of simpler, more analytically and computationally tractable models in learning tasks. 

\paragraph{Overview.}
This paper constructs a bridge from biologically realistic to analytically tractable models. The \emph{selectron} is a model derived from leaky integrate and fire neurons equipped with spike-timing dependent plasticity that is amenable to learning-theoretic analysis. Our aim is to extract some of the principles implicit in STDP by thoroughly investigating a limit case. 

Section \S\ref{s:selectron} introduces the selectron. We state a constrained reward maximization problem which implies that selectrons encode empirical reward estimates into spikes. Our first result, section \S\ref{s:relations}, is that the selectron arises as the fast time constant limit of well-established models of neuronal spiking and plasticity, suggesting that cortical neurons may also be encoding reward estimates into their spiketrains. 

Two important questions immediately arise. First, what guarantees can be provided on spikes being reliable predictors of global (neuromodulatory) outcomes? Second, what guarantees can be provided on the usefulness of spikes to other neurons? Sections \S\ref{s:learning} and \S\ref{s:code} answer these questions by providing an upper bound on a suitably defined $0/1$ loss and a lower bound on the \emph{efficacy} of a selectron's spikes, measured in terms of its contribution to the expected reward of a downstream selectron. Both bounds are controlled by the sum of synaptic weights $\|\wt\|_1$, thereby justifying the constraint introduced in \S\ref{s:selectron}. Finally, motivated by our analysis, \S\ref{s:experiments} introduces a \emph{regularized STDP} rule and shows that it learns more robustly than classical STDP. \S\ref{s:discussion} concludes the paper. Proofs of theorems are provided in the supplementary material.

\paragraph{Related work.}
Spike-timing dependent plasticity and its implications for the neural code have been intensively studied in recent years. The work closest in spirit to our own is Seung's ``hedonistic'' synapses, which seek to increase average reward \cite{seung:03}. Our work provides guarantees on the finite sample behavior of a discrete-time analog of hedonistic neurons. Another related line of research derives from the information bottleneck method \cite{tishby:99, buesing:07} which provides an alternate constraint to the one considered here. An information-theoretic perspective on synaptic homeostasis and metabolic cost, complementing the results in this paper, can be found in \cite{bt:12, bob:12}. Simulations combining synaptic renormalization with burst-STDP can be found in \cite{Nere:2012fk}.

Important aspects of plasticity that we have not considered here are properties specific to continuous-time models, such as STDP's behavior as a temporal filter \cite{schmiedt:10}, and also issues related to convergence \cite{legenstein:06, legenstein:08a}.

The learning-theoretic properties of neural networks have been intensively studied, mostly focusing on perceptrons, see for example \cite{anthony:99}. A non-biologically motivated ``large-margin'' analog of the perceptron was proposed in \cite{freund:99a}.

\section{The selectron}
\label{s:selectron}

We introduce the selectron, which can be considered a biologically motivated adaptation of the perceptron, see \S\ref{s:relations}. The mechanism governing whether or not the selectron spikes is a Heaviside function acting on a weighted sum of synaptic inputs; our contribution is to propose a new reward function and corresponding learning rule. 

Let us establish some notation. Let $\X$ denote the set of $N$-dimensional $\{0,1\}$-valued vectors forming synaptic inputs to a selectron, and $Y=\{0,1\}$ the set of outputs. A selectron spikes according to 
\begin{equation}
	\label{e:threshold}
	y=f_{\wt}(\x) := H\left(\wt^\intercal \x-\vartheta\right),
	\text{ where }
	H(z):= \begin{cases}
		1 & \text{if }z>0\\
		0 & \text{else}
	\end{cases}
\end{equation}
is the Heaviside function and $\wt$ is a $[0,1]\subset \bR$ valued $N$-vector specifying the selectron's synaptic weights. Let $P(\x)$ denote the probability of input $\x$ arising.

To model the neuromodulatory system we introduce random variable $\nu:\X\rightarrow\{-1,0,+1\}$, where positive values correspond to desirable outcomes, negative to undesirable and zero to neutral. Let $P(\nu|\x)$ denote the probability of the release of neuromodulatory signal subsequent to input $\x$. 

\begin{defn}
	Define \textbf{reward function}
\begin{equation}
	\label{e:reward}
	R(\x,f_{\wt},\nu) 
	= \underbrace{\nu(\x)}_{\text{neuromodulators}}\cdot \underbrace{(\wt^\intercal \x-\vartheta)}_{\text{margin}}\cdot \underbrace{f_{\wt}(\x)}_{\text{selectivity}} =
	\begin{cases}
		\nu(\x)\cdot (\wt^\intercal \x-\vartheta) & \text{if }y=1\\
		0 & \text{else}.
	\end{cases}
\end{equation}
\end{defn}
The reward consists in three components. The first term is the neuromodulatory signal, which acts as a supervisor. The second term is the total current $\wt^\intercal\x$ minus the threshold $\vartheta$. It is analogous to the margin in support vector machines or boosting algorithms, see section \S\ref{s:learning} for a precise formulation. 

The third term gates rewards according to whether or not the selectron spikes. The reward is thus \emph{\textbf{selected}}\footnote{The name ``selectron'' was chosen to emphasize this selective aspect.}: neuromodulatory signals are ignored by the selectron's reward function when it does not spike, enabling specialization.
 
\paragraph{Constrained reward maximization.}
The selectron solves the following optimization problem:
\begin{align}
	\underset{\wt}{\text{maximize:}} \hspace{3mm} & 
	\widehat{R}_n:=\sum_{i=1}^n \nu(\x^{(i)})\cdot (\wt^\intercal \x^{(i)}-\vartheta)\cdot f_{\wt}(\x^{(i)})
	\label{e:goal}\\
	\mbox{subject to:} \hspace{3mm} &  \|\wt\|_1\leq\omega
	\,\,\text{ for some $\omega>0$.}
	\notag 
\end{align}	

\begin{rem}[spikes encode rewards]\eod
Optimization problem \eqref{e:goal} ensures that selectrons spike for inputs that, on the basis of their empirical sample, reliably lead to neuromodulatory rewards. 
Thus, spikes encode expectations about rewards.
\end{rem}

The constraint is motivated by the discussion after Theorem~\ref{t:frequency} and the analysis in \S\ref{s:learning} and \S\ref{s:code}. We postpone discussion of how to impose the constraint to \S\ref{s:experiments}, and focus on reward maximization here.

The reward maximization problem cannot be solved analytically in general. However, it is possible to use an iterative approach. Although $f_\wt(\x)$ is not continuous, the reward function is a continuous function of $\wt$ and is differentiable everywhere except for the ``corner'' where $\wt^\intercal\x-\vartheta=0$. We therefore apply gradient ascent by computing the derivative of \eqref{e:goal} with respect to synaptic weights to obtain online learning rule
\begin{equation}
	\Delta \wt_{j} = \alpha\cdot \nu(\x)\cdot \x_j\cdot f_{\wt}(\x) 
	= \begin{cases}
	\alpha\cdot \nu(\x) & \mbox{if }\x_j=1 \text{ and } y=1
	\\
	0 &\mbox{else}
	\end{cases}
	\label{e:grad-asc}
\end{equation}
where update factor $\alpha$ controls the learning rate.

The learning rule is \emph{\textbf{selective}}: regardless of the neuromodulatory signal, synapse $\wt_{jk}$ is updated only if there is both an input $\x_j=1$ and output spike $y=f_{\wt}(\x)=1$.

The selectron is not guaranteed to find a global optimum. It is prone to initial condition dependent local optima because rewards depend on output spikes in learning rule \eqref{e:grad-asc}. Although this is an undesirable property for an isolated learner, it is less important, and perhaps even advantageous, in large populations where it encourages specialization. 

\begin{rem}[unsupervised setting]\label{eg:saturate}\eod
	Define the \emph{unsupervised setting} by $\nu(\x)=1$ for all $\x$. The reward function reduces to $R(\x, f_{\wt})=(\wt^\intercal \x-\vartheta)\cdot f_{\wt}(\x)$. Without the constraint synapses will saturate. Imposing the constraint yields a more interesting solution where the selectron finds a weight vector summing to $\omega$ which balances (i) frequent spikes and (ii) high margins.	
\end{rem}

\begin{thm}[Controlling the frequency of spikes]\label{t:frequency}\eod
	Assuming synaptic inputs are i.i.d. Bernoulli variables with $P(\text{spike})=p$, then
	\begin{equation*}
		P\Big(f_\wt(\x)=1\Big) \leq p\cdot\left(\frac{\|\wt\|_1}{\vartheta}\right)^2
		\leq p\cdot\left(\frac{\omega}{\vartheta}\right)^2. 
	\end{equation*}
\end{thm}
The Bernoulli regime is the discrete-time analog of the homogeneous Poisson setting used to prove convergence of  reward-modulated STDP in \cite{legenstein:06}. Interestingly, in this setting the constraint provides a lever for controlling (lower bounding) rewards per spike
\begin{equation*}
	\Big\{\text{reward per spike}\Big\}
	= \frac{\widehat{R}}{P(f_\wt(\x)=1)}
	 \geq c_1\cdot \frac{\widehat{R}}{\omega^2}.
\end{equation*}
If inputs are not Bernoulli i.i.d., then $P(y=1)$ and $\omega$ still covary, although the precise relationship is more difficult to quantify. Although i.i.d. inputs are unrealistic, note that recent neurophysiological evidence suggests neuronal firing -- even of nearby neurons -- is uncorrelated \cite{Ecker:2010fk}.

\section{Relation to leaky integrate-and-fire neurons equipped with STDP}
\label{s:relations}

The literature contains an enormous variety of neuronal models, which vary dramatically in sophistication and the extent to which they incorporate the the details of the underlying biochemical processes. Similarly, there is a large menagerie of models of synaptic plasticity \cite{dan:04}. We consider two well-established models: Gerstner's Spike Response Model (SRM) which generalizes leaky integrate-and-fire neurons \cite{gerstner:95} and the original spike-timing dependent plasticity  learning rule proposed by Song \emph{et al} \cite{song:00}, and show that the selectron arises in the fast time constant limit of the two models.

First let us recall the SRM. Suppose neuron $n^k$ last outputted a spike at time $t_k$ and receives input spikes at times $t_j$ from neuron $n^j$. Neuron $n^k$ spikes or according to the Heaviside function applied to the membrane potential $M_\wt$:
\begin{equation*}
	f_{\wt}(t) = H\left(M_{\wt}(t)-\vartheta\right)
	\text{ where }
	M_{\wt}(t) = \eta(t-t_k) + \sum_{t_j\leq t}\wt_{jk}\cdot\epsilon(t-t_j)
	\text{ at time $t\geq t_k$}.
\end{equation*}
Input and output spikes add 
\begin{gather*}
	\epsilon(t-t_j) = K\left[e^{\left(\frac{t_j-t}{\tau_m}\right)}
	-e^{\left(\frac{t_j-t}{\tau_s}\right)}\right]
	\text{ and }
	\eta(t-t_k) = \vartheta\left[K_1 e^{\left(\frac{t_k-t}{\tau_m}\right)}
	-K_2\left(e^{\left(\frac{t_k-t}{\tau_m}\right)}
	-e^{\left(\frac{t_k-t}{\tau_s}\right)}\right)\right]
\end{gather*}
to the membrane potential for $t_j\leq t$ and $t_k\leq t$ respectively. Here $\tau_m$ and $\tau_s$ are the membrane and synapse time constants. 

The original STDP update rule \cite{song:00} is 
\begin{equation}
	\label{e:stdp}
	\Delta \wt_{jk} =\begin{cases}
 	\alpha_+\cdot e^{\left(\frac{t_j-t_k}{\tau_+}\right)} & \mbox{ if }t_j\leq t_k
	\\  
	-\alpha_-\cdot e^{\left(\frac{t_k-t_j}{\tau_-}\right)} & \mbox{ else}
	\end{cases}
\end{equation}
where $\tau_+$ and $\tau_-$ are time constants. STDP potentiates input synapses that spike prior to output spikes and depotentiates input synapses that spike subsequent to output spikes.

\begin{thm}[the selectron is the fast time constant limit of SRM + STDP]\label{t:limit}\eod
	In the fast time constant limit, $\lim_{\tau_\bullet}\rightarrow 0$, the SRM transforms into a selectron with
	\begin{equation*}
		f_\wt(t) = H\Big(M_\wt(t)-\vartheta\Big)
		\quad\text{ where }\quad M_\wt = \sum_{\{j|t_j\geq t_k\}}\wt_{jk}\cdot \delta_{t_k}(t).
	\end{equation*}
	Moreover, STDP transforms into learning rule \eqref{e:grad-asc} in the unsupervised setting with $\nu(\x)=1$ for all $\x$. Finally, STDP arises as gradient ascent on a reward function whose limit is the unsupervised setting of reward function $\eqref{e:reward}$.
\end{thm}

Theorem~\ref{t:limit} shows that STDP implicitly maximizes a time-discounted analog of the reward function in \eqref{e:goal}. We expect many models of reward-modulated synaptic plasticity to be analytically tractable in the fast time constant limit.  An important property  shared by STDP and the selectron is that synaptic (de)potentiation is gated by output spikes, see \S%
\ref{s:perceptron} for a comparison with the perceptron which \emph{does not} gate synaptic learning

\section{An error bound}
\label{s:learning}

Maximizing reward function \eqref{e:goal} implies that selectrons encode reward estimates into their spikes. Indeed, it recursively justifies incorporating spikes into the reward function via the margin $(\wt^\intercal \x-\vartheta)$, which only makes sense if upstream spikes predict reward. However, in a large system where estimates pile on top of each other there is a tendency to \emph{overfit}, leading to poor generalizations \cite{geman:92}. It is therefore crucial to provide \emph{guarantees} on the quality of spikes as estimators. 

Boosting algorithms, where the outputs of many weak learners are aggregated into a classifier \cite{freund:96}, are remarkably resistant to overfitting as the number of learners increases \cite{schapire:98}. Cortical learning may be analogous to boosting: individual neurons have access to a tiny fraction of the total brain state, and so are weak learners; and in the fast time constant limit, neurons are essentially aggregators. 

We sharpen the analogy using the selectron. As a first step towards understanding how the cortex combats overfitting, we adapt a theorem developed to explain the effectiveness of boosting \cite{boucheron:05}. The goal is to show how the margin and constraint on synaptic weights improve error bounds.

\begin{defn}\label{d:z1}
	A selectron incurs a \textbf{$0/1$ loss} if a spike is followed by negative neuromodulatory feedback
\begin{equation}
	\label{e:z1loss}
	l(\x,f_\wt,\nu)=\bO_{-f_\wt(\x)\cdot \nu(\x)}=\begin{cases}
		1 & \text{if }y=1\text{ and }\nu(\x)=-1\\
		0 & \text{else.}
	\end{cases}
\end{equation}
The $0/1$ loss fails to take the estimates (spikes) of other selectrons into account and is difficult to optimize, so we also introduce the \textbf{hinge loss}:
\begin{equation}
	h^\kappa(\x,f_{\wt},\nu) 
	:= \Big(\kappa-(\wt^\intercal \x-\vartheta)\cdot \nu(\x)\Big)_+\cdot f_{\wt}(\x),
	\text{ where }
	(x)_+:=\begin{cases}
		x & \text{if }x\geq0\\
		0 & \text{else}.
	\end{cases}
\end{equation}
Note that $l\leq h^\kappa$ for all $\kappa\geq 1$. Parameter $\kappa$ controls the saturation point, beyond which the size of the margin makes no difference to $h^\kappa$.
\end{defn}
An alternate $0/1$ loss\footnote{See \S%
\ref{s:hmargin} for an error bound.} penalizes a selectron if it {\rm (i)} fires when it shouldn't, i.e. when $\nu(\x)=-1$ or {\rm (ii)} does not fire when it should, i.e. when $\nu(\x)=1$.  However, since the cortex contains many neurons and spiking is metabolically expensive \cite{Hasenstaub:2010fk}, we propose a conservative loss that only penalizes errors of commission (``first, do no harm'') and does not penalize specialization.

\begin{thm}[spike error bound]\label{t:lbound}\eod
	Suppose each selectron has $\leq N$ synapses. For any selectron $n^k$, let $S^k=\{n^k\}\cup \{n^j:n^j\rightarrow n^k\}$ denote a 2-layer feedforward subnetwork. 
	For all $\kappa\geq 1$, with probability at least $1-\delta$,
	\begin{align*}
		\bE\underbrace{\big[l(\x,f_\wt,\nu)\big]}_{0/1\text{ loss}}
		\leq & \frac{1}{n}\sum_i \underbrace{h^\kappa\big(\x^{(i)},f_\wt,\nu(\x^{(i)})\big)}_{\text{hinge loss}}
		+ \omega\cdot\underbrace{2B\cdot
		\frac{\sqrt{8(N+1)\log(n+1)}+1}{\sqrt{n}}}_{\text{capacity term}}\\
		& + 2B\cdot\underbrace{\sqrt{\frac{2\log\frac{2}{\delta}}{n}}}_{\text{confidence term}}
		\hspace{15mm}\text{where }B=\kappa+\omega-\vartheta.
	\end{align*}
\end{thm}

\begin{rem}[theoretical justification for maximizing margin and constraining $\|\wt\|_1$]\eod
	The theorem shows how subsets of distributed systems can avoid overfitting. First, it demonstrates the importance of maximizing the margin (i.e. the empirical reward). Second, it shows the capacity term depends on the number of synapses $N$ and the constraint $\omega$ on synaptic weights, rather than the capacity of $S^k$ -- which can be very large. 
\end{rem}

The hinge loss is difficult to optimize directly since gating with output spikes $f_\wt(\x)$ renders it discontinuous. However, in the Bernoulli regime, Theorem~\ref{t:frequency} implies the bound in Theorem~\ref{t:lbound} can be rewritten as
\begin{equation}
	\label{e:rew-error}
	\bE\big[l(\x,f_\wt,\nu)\big]
	\leq p\kappa\frac{\omega^2}{\vartheta^2} 
	- \widehat{R}_n\big(\x^{(i)},f_\wt,\nu(\x^{(i)})\big)
	+ \omega\cdot\big\{\text{capacity term}\big\}
	+ \big\{\text{confidence term}\big\}
\end{equation}
and so $\omega$ again provides the lever required to control the $0/1$ loss. The constraint $\|\wt\|_1\leq\omega$ is best imposed offline, see \S\ref{s:experiments}.

\section{A bound on the efficacy of inter-neuronal communication}
\label{s:code}

 Even if a neuron's spikes perfectly predict positive neuromodulatory signals, the spikes only matter to the extent they affect other neurons in cortex. Spikes are produced for neurons by neurons. It is therefore crucial to provide guarantees on the usefulness of spikes.

In this section we quantify the effect of one selectron's spikes on another selectron's expected reward. We demonstrate a lower bound on efficacy and discuss its consequences.

\begin{defn}
	The \textbf{efficacy} of spikes from selectron $n^j$ on selectron $n^k$ is
	\begin{equation*}
		\frac{\delta R^k}{\delta \x_j} 
		:= \frac{\bE[R^k|\x_j=1] - \bE[R^k|\x_j=0]}{1-0},
	\end{equation*}
	i.e. the expected contribution of spikes from selectron $n^j$ to selectron $n^k$'s expected reward, relative to not spiking. The notation is intended to suggest an analogy with differentiation -- the infinitesimal difference made by spikes on a single synapse. 
\end{defn}

Efficacy is zero if $\bE[R^k|\x_j=1] = \bE[R^k|\x_j=0]$. In other words, if spikes from $n^j$ make no difference to the expected reward of $n^k$. 

The following theorem relies on the assumption that the average contribution of neuromodulators is higher after $n^j$ spikes than after it does not spike (i.e. upstream spikes predict reward), see \S%
\ref{s:composition-proof} for precise statement. When the assumption is false the synapse $\wt_{jk}$ should be pruned.

\begin{thm}[spike efficacy bound]\label{t:composition}\eod
	Let $p_j:=\bE[Y^j]$ denote the frequency of spikes from neuron $n^j$. The efficacy of $n^j$'s spikes on $n^k$ is lower bounded by
	\begin{equation}
		\label{e:rcomp}
		c_2\cdot\underbrace{\frac{\delta R^k}{\delta \x_j}}_{\text{efficacy}}
		\geq \underbrace{\frac{\wt_j\cdot\bE[Y^jY^k]}{p_j}}_{\wt_j \text{-weighted co-spike frequency}}
		+ \underbrace{\frac{2\bE\Big[Y^jY^k\cdot \big((\wt^{\bcancel j})^\intercal \x-\vartheta\big)\Big]}{p_j(1-p_j)}}_{\text{co-spike frequency}}
		- \underbrace{\frac{\bE\Big[Y^k\cdot \big((\wt^{\bcancel j})^\intercal \x-\vartheta\big)\Big]}{1-p_j}}_{n^k\text{ spike frequency}}
	\end{equation}	
	where $c_2$ is described in \S%
\ref{s:composition-proof} and $\wt^{\bcancel j}_i := \wt_i$ if $i\neq j$ and $0$ if $i=j$.
\end{thm}
The efficacy guarantee is interpreted as follows. First, the guarantee improves as co-spiking by $n^j$ and $n^k$ increases. However, the denominators imply that increasing the frequency of $n^j$'s spikes \emph{worsens} the guarantee, insofar as $n^j$ is not correlated with $n^k$. Similarly, from the third term, increasing $n^k$'s spikes \emph{worsens} the guarantee if they do not correlate with $n^j$.

An immediate corollary of Theorem~\ref{t:composition} is that Hebbian learning rules, such as STDP and the selectron learning rule \eqref{e:grad-asc}, improve the efficacy of spikes. However, it also shows that naively increasing the frequency of spikes carries a cost. Neurons therefore face a tradeoff. In fact, in the Bernoulli regime, Theorem~\ref{t:frequency} implies \eqref{e:rcomp} can be rewritten as
\begin{equation}
	\label{e:rcomp2}
	c_2\cdot\frac{\delta R^k}{\delta \x_j}
	\geq \frac{\wt_j}{p}\cdot \bE[Y^jY^k]
	+ \frac{2}{p(1-p)}\bE\Big[Y^jY^k\cdot \big((\wt^{\bcancel j})^\intercal \x-\vartheta\big)\Big]
	- \frac{p\cdot\omega^2\cdot(\omega-\vartheta)}{(1-p)\vartheta^2},
\end{equation}
so the constraint 
$\omega$  on synaptic strength can be used as a lever to improve guarantees on efficacy.

\begin{rem}[efficacy improved by pruning weak synapses]\eod
	\label{r:prune}
	The $1^{st}$ term in \eqref{e:rcomp} suggests that pruning weak synapses increases the efficacy of spikes, and so may aid learning in populations of selectrons or neurons.
\end{rem}

\section{Experiments}
\label{s:experiments}

Cortical neurons are constantly exposed to different input patterns as organisms engage in different activities. It is therefore important that what neurons learn is robust to changing inputs \cite{fusi:05, fusi:07}. In this section, as proof of principle, we investigate a simple tweak of classical STDP involving offline regularization. We show that it improves robustness when neurons are exposed to more than one pattern.

Observe that \emph{regularizing} optimization problem \eqref{e:goal} yields
\begin{align}
	\underset{\wt}{\text{maximize:}} \hspace{3mm} &
	\sum_{i=1}^n R\big(\x^{(i)},f_{\wt},\nu(\x^{(i)})\big) - \frac{\gamma}{2}(\|\wt\|_1-\omega)^2\\
	\label{e:constrained-ascent}
	\text{learning rule:}\hspace{3mm} &
	\Delta \wt_{j} = \alpha\cdot \nu(\x)\cdot \x_j\cdot f_{\wt}(\x) -\gamma\cdot\big(\|\wt\|_1-\omega\big)\cdot\wt_j
\end{align}
incorporates synaptic renormalization directly into the update. However, \eqref{e:constrained-ascent} requires continuously re-evaluating the sum of synaptic weights. We therefore \emph{decouple} learning into an online reward maximization phase and an offline regularization phase which resets the synaptic weights.

A similar decoupling may occur in cortex. It has recently been proposed that a function of NREM sleep may be to regulate synaptic weights \cite{Tononi:06}. Indeed, neurophysiological evidence suggests that average cortical firing rates increase during wakefulness and decrease during sleep, possibly reflecting synaptic strengths \cite{vyazovskiy:08, vyazovskiy:09}. Experimental evidence also points to a net increase in dendritic spines (synapses) during waking and a net decrease during sleep \cite{Maret:2011rt}. 

\paragraph{Setup.}
We trained a neuron on a random input pattern for 10s to 87\% accuracy with regularized STDP. See \S%
\ref{s:poisson} for details on the structure of inputs. We then performed 700 trials (350 classical and 350 regularized) exposing the neuron to a new pattern for 20 seconds and observed performance under classical and regularized STDP. 


\paragraph{SRM neurons with classical STDP.}
We used Gerstner's SRM model, recall \S\ref{s:relations}, with parameters chosen to exactly coincide with \cite{Masquelier:2007vn}: $\tau_m=10$, $\tau_s=2.5$, $K=2.2$, $K_1=2$, $K_2=4$ and $\vartheta=\frac{1}{4}\#\text{synapses}$.
STDP was implemented via \eqref{e:stdp} with parameters $\alpha_+=0.03125$, $\tau_+=16.8$, $\alpha_-=0.85\alpha_+$ and $\tau_-=33.7$ also taken from \cite{Masquelier:2007vn}. Synaptic weights were clipped to fall in $[0,1]$.

\paragraph{Regularized STDP}
consists of a small tweak of classical STDP in the online phase, and an additional offline regularization phase:
\begin{itemize}
	\item \emph{Online.} In the online phase, reduce the depotentiation bias from $0.85\alpha_+$ in the classical implementation to $\alpha_-=0.75\alpha_+$.
	\item \emph{Offline.}
	In the offline phase, modify synapses once per second according to 
	\begin{equation}
		\Delta \wt_{j} = \begin{cases}
			\gamma\cdot \left(\frac{3}{2}-\wt_{j}\right)\cdot (\omega - s) & \mbox{ if } \omega< s \\ 
			\gamma\cdot (\omega - s)& \mbox{ else,}
		\end{cases}
		\label{e:sstdp}
	\end{equation}
	where $s$ is output spikes per second, $\omega=5Hz$ is the target rate and update factor $\gamma=0.6$. The offline update rule is firing rate, and not spike, dependent.
\end{itemize}

Classical STDP has a depotentiation bias to prevent runaway potentiation feedback loops leading to seizures \cite{song:00}. Since synapses are frequently renormalized offline we incorporate a weak exploratory (potentiation) bias during the online phase which helps avoid local minima.\footnote{The input stream contains a repeated pattern, so there is a potentiation bias in practice even though the net integral of STDP in the online phase is negative.} This is in line with experimental evidence showing increased cortical activity during waking \cite{vyazovskiy:09}.

Since computing the sum of synaptic weights is non-physiological, we draw on Theorem~\ref{t:frequency} and use the neuron's firing rate when responding to uncorrelated inputs as a proxy for $\|\wt\|_1$. Thus, in the offline phase, synapses receive inputs generated as in the online phase but without repeated patterns. Note that \eqref{e:constrained-ascent} has a larger pruning effect on stronger synapses, discouraging specialization. Motivated by Remark~\ref{r:prune}, we introduce bias $(\frac{3}{2}-\wt_j)$ in the offline phase to ensure weaker synapses are downscaled more than strong synapses. For example, a synapse with $\wt_{i}=0.5$ is downscaled by \emph{twice} as much as a synapse with weight $\wt_{j}=1.0$.

Regularized STDP alternates between 2 seconds online and 4 seconds offline, which suffices to renormalize synaptic strengths. The frequency of the offline phase could be reduced by decreasing the update factors $\alpha_\pm$, presenting stimuli less frequently (than 7 times per second), or adding inhibitory neurons to the system. 

\paragraph{Results.}
A summary of results is presented in the table below: accuracy quantifies the fraction of spikes that co-occur with each pattern. Regularized STDP outperforms classical STDP on both patterns on average. It should be noted that regularized neurons were not only online for $20$ seconds but also offline -- and exposed to Poisson noise -- for $40$ seconds. Interestingly, exposure to Poisson noise improves performance.

\begin{center}\begin{small}
    \begin{tabular*}{.35\textwidth}{c c c}
	\toprule
		\textbf{Algorithm} &  \multicolumn{2}{c}{\textbf{Accuracy}}\\ & Pattern 1 & Pattern 2 \\		
	\midrule
	Classical &  54\% & 39\% \\
	\midrule
	Regularized  &  59\% & 48\% \\
	\bottomrule
    \end{tabular*}
\end{small}\end{center}

Fig.~\ref{f:unlearn} provides a more detailed analysis.  Each panel shows a 2D-histogram (darker shades of gray correspond to more trials) plotting accuracies on both patterns simultaneously, and two 1D histograms plotting accuracies on the two patterns separately. The 1D histogram for regularized STDP shows a unimodal distribution for pattern \#2, with most of the mass over accuracies of 50-90\%. For pattern \#1, which has been ``unlearned'' for twice as long as the training period, most of the mass is over accuracies of 50\% to 90\%, with a significant fraction ``unlearnt''. By contrast, classical STDP exhibits extremely brittle behavior. It completely unlearns the original pattern in about half the trials, and also fails to learn the new pattern in most of the trials.

Thus,  as suggested by our analysis, introducing a regularization both improves the robustness of STDP and enables an exploratory bias by preventing runaway feedback leading to epileptic seizures.

\begin{figure}[t]
	\centering
	\subfigure[Classical STDP]
	{\includegraphics[width=0.49\textwidth]{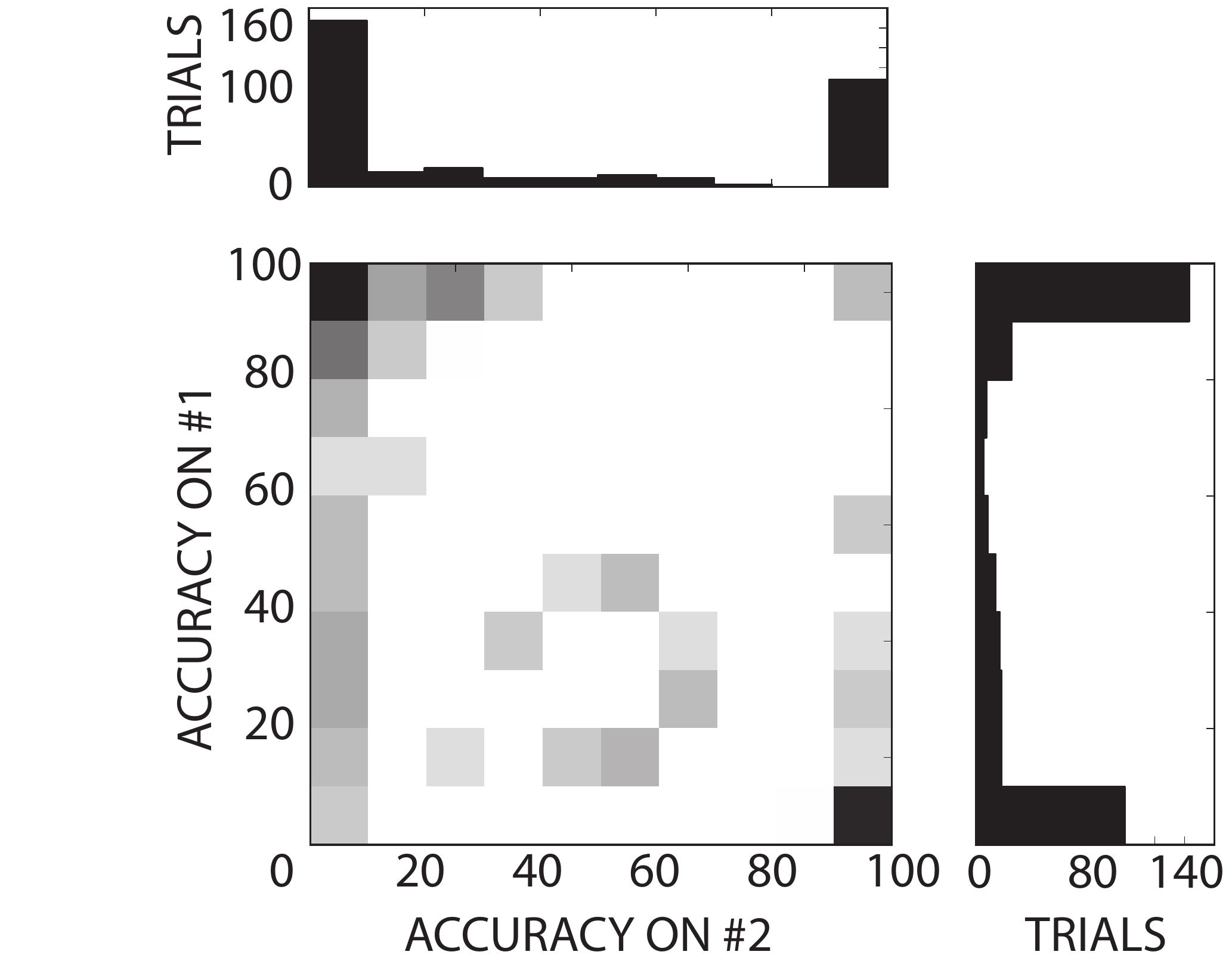} 
	\label{f:hit_c}}
	\subfigure[Regularized STDP]
	{\includegraphics[width=0.49\textwidth]{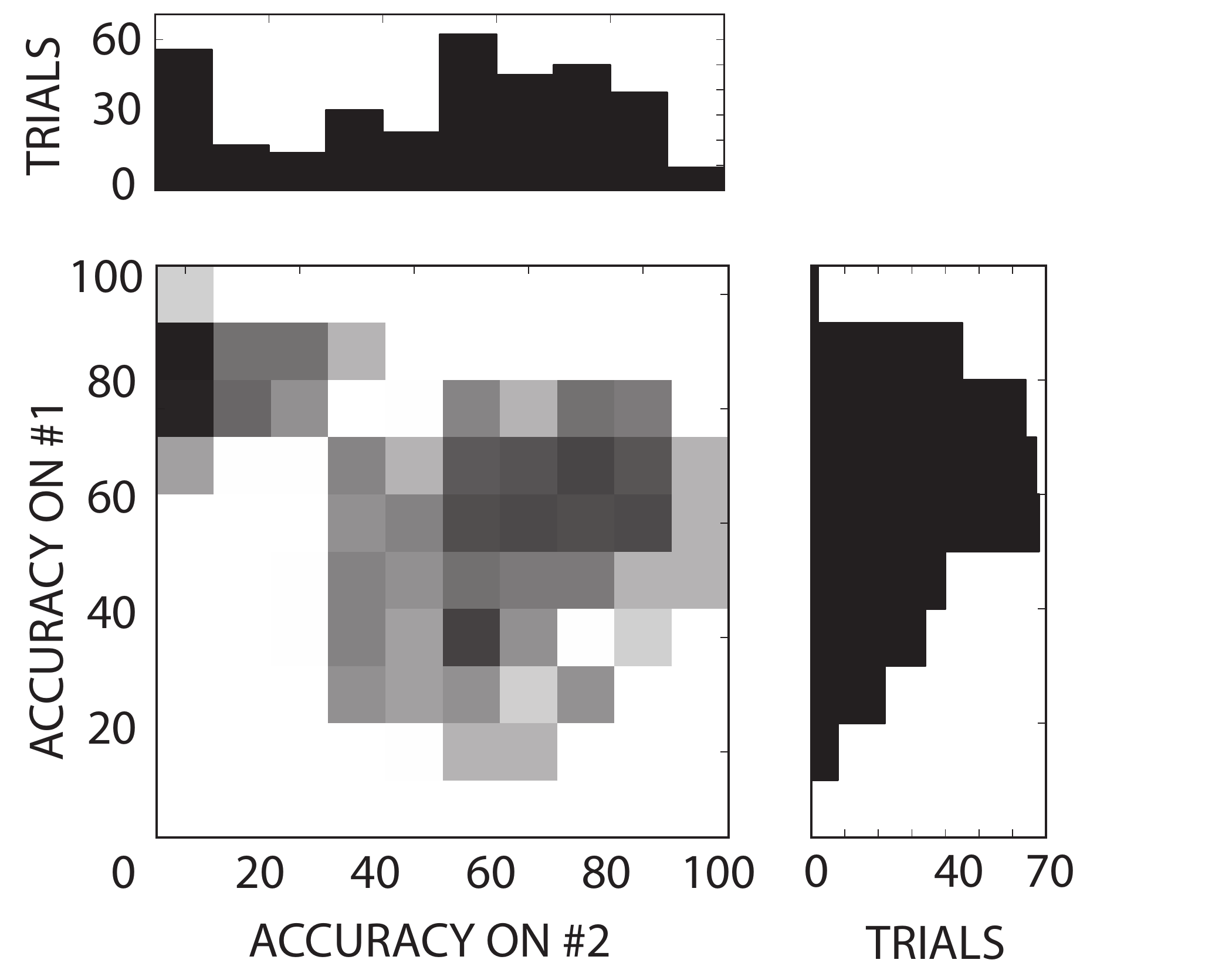} 
	\label{f:hit_w}}
		\vspace{-3mm}
	\caption{Accuracy after 20 seconds of exposure to a novel pattern.
	}
	\label{f:unlearn}
\end{figure}

\section{Discussion}
\label{s:discussion}

The selectron provides a bridge between a particular model of spiking neurons -- the Spike Response Model \cite{gerstner:95} with the original spike-timing dependent plasticity rule \cite{song:00} -- and models that are amenable to learning-theoretic analysis.  Our hope is that the selectron and related models lead to an improved understanding of the principles underlying learning in cortex. It remains to be seen whether other STDP-based models also have tractable discrete-time analogs. 

The selectron is an interesting model in its own right: it embeds reward estimates into spikes and maximizes a margin that improves error bounds. It imposes a constraint on synaptic weights that: concentrates rewards/spike, tightens error bounds and improves guarantees on spiking efficacy. Although the analysis does not apply directly to continuous-time models, experiments show that a tweak inspired by our analysis improves the performance of a more realistic model. An important avenue for future research is investigating the role of feedback in cortex, specifically NMDA synapses, which may have interesting learning-theoretic implications.

\textbf{Acknowledgements.}
We thank Timoth{\'e}e Masquelier for generously sharing his source code \cite{Masquelier:2007vn} and Samory Kpotufe for useful discussions.

{\small

\newcommand{\BMCxmlcomment}[1]{}

\BMCxmlcomment{

<refgrp>

<bibl id="B1">
  <title><p>A free energy principle for the brain</p></title>
  <aug>
    <au><snm>Friston</snm><fnm>KJ</fnm></au>
    <au><snm>Kilner</snm><fnm>J</fnm></au>
    <au><snm>Harrison</snm><fnm>L</fnm></au>
  </aug>
  <source>J. Phys. Paris</source>
  <pubdate>2006</pubdate>
  <volume>100</volume>
  <fpage>70</fpage>
  <lpage>-87</lpage>
</bibl>

<bibl id="B2">
  <title><p>The perceptron: a probabilistic model for information storage and
  organization in the brain</p></title>
  <aug>
    <au><snm>Rosenblatt</snm><fnm>F</fnm></au>
  </aug>
  <source>Psychol Rev</source>
  <pubdate>1958</pubdate>
  <volume>65</volume>
  <issue>6</issue>
  <fpage>386</fpage>
  <lpage>408</lpage>
</bibl>

<bibl id="B3">
  <title><p>Learning representations by back-propagating errors</p></title>
  <aug>
    <au><snm>Rumelhart</snm><fnm>D E</fnm></au>
    <au><snm>Hinton</snm><fnm>G E</fnm></au>
    <au><snm>Williams</snm><fnm>R J</fnm></au>
  </aug>
  <source>Nature</source>
  <pubdate>1986</pubdate>
  <volume>323</volume>
  <fpage>533</fpage>
  <lpage>536</lpage>
</bibl>

<bibl id="B4">
  <title><p>A {F}ast {L}earning {A}lgorithm for {D}eep {B}elief
  {N}ets</p></title>
  <aug>
    <au><snm>Hinton</snm><fnm>GE</fnm></au>
    <au><snm>Osindero</snm><fnm>S</fnm></au>
    <au><snm>Teh</snm><fnm>Y W</fnm></au>
  </aug>
  <source>Neural Computation</source>
  <pubdate>2006</pubdate>
  <volume>18</volume>
  <fpage>1527</fpage>
  <lpage>1554</lpage>
</bibl>

<bibl id="B5">
  <title><p>Competitive {H}ebbian learning through spike-timing-dependent
  synaptic plasticity</p></title>
  <aug>
    <au><snm>Song</snm><fnm>S</fnm></au>
    <au><snm>Miller</snm><fnm>K D</fnm></au>
    <au><snm>Abbott</snm><fnm>L F</fnm></au>
  </aug>
  <source>Nature Neuroscience</source>
  <pubdate>2000</pubdate>
  <volume>3</volume>
  <issue>9</issue>
</bibl>

<bibl id="B6">
  <title><p>Learning in {S}piking {N}eural {N}etworks by {R}einforcement of
  {S}tochastic {S}ynaptic {T}ransmission</p></title>
  <aug>
    <au><snm>Seung</snm><fnm>H S</fnm></au>
  </aug>
  <source>Neuron</source>
  <pubdate>2003</pubdate>
  <volume>40</volume>
  <issue>1063-1073</issue>
</bibl>

<bibl id="B7">
  <title><p>Reducing spike train variability: {A} computational theory of
  spike-timing dependent plasticity</p></title>
  <aug>
    <au><snm>Bohte</snm><fnm>S M</fnm></au>
    <au><snm>Mozer</snm><fnm>M C</fnm></au>
  </aug>
  <source>Advances in Neural Information Processing Systems (NIPS)</source>
  <pubdate>2005</pubdate>
</bibl>

<bibl id="B8">
  <title><p>A criterion for the convergence of learning with spike timing
  dependent plasticity</p></title>
  <aug>
    <au><snm>Legenstein</snm><fnm>R</fnm></au>
    <au><snm>Maass</snm><fnm>W</fnm></au>
  </aug>
  <source>Advances in Neural Information Processing Systems (NIPS)</source>
  <pubdate>2006</pubdate>
</bibl>

<bibl id="B9">
  <title><p>Simplified rules and theoretical analysis for information
  bottleneck optimization and {P}{C}{A} with spiking neurons</p></title>
  <aug>
    <au><snm>Buesing</snm><fnm>L</fnm></au>
    <au><snm>Maass</snm><fnm>W</fnm></au>
  </aug>
  <source>Adv in Neural Information Processing Systems (NIPS)</source>
  <pubdate>2007</pubdate>
</bibl>

<bibl id="B10">
  <title><p>Theoretical analysis of learning with reward-modulated
  spike-timing-dependent plasticity</p></title>
  <aug>
    <au><snm>Legenstein</snm><fnm>R</fnm></au>
    <au><snm>Pecevski</snm><fnm>D</fnm></au>
    <au><snm>Maass</snm><fnm>W</fnm></au>
  </aug>
  <source>Advances in Neural Information Processing Systems (NIPS)</source>
  <pubdate>2008</pubdate>
</bibl>

<bibl id="B11">
  <title><p>The information bottleneck method</p></title>
  <aug>
    <au><snm>Tishby</snm><fnm>N</fnm></au>
    <au><snm>Pereira</snm><fnm>FC</fnm></au>
    <au><snm>Bialek</snm><fnm>W</fnm></au>
  </aug>
  <source>Proc. of the 37-th Annual Allerton Conference on Communication,
  Control and Computing</source>
  <editor>B Hajek and RS Sreenivas</editor>
  <pubdate>1999</pubdate>
</bibl>

<bibl id="B12">
  <title><p>What can neurons do for their brain? {C}ommunicate selectivity with
  spikes</p></title>
  <aug>
    <au><snm>Balduzzi</snm><fnm>D</fnm></au>
    <au><snm>Tononi</snm><fnm>G</fnm></au>
  </aug>
  <source>To appear in Theory in Biosciences</source>
  <pubdate>2012</pubdate>
</bibl>

<bibl id="B13">
  <title><p>Metabolic cost as an organizing principle for cooperative
  learning</p></title>
  <aug>
    <au><snm>Balduzzi</snm><fnm>D</fnm></au>
    <au><snm>Ortega</snm><fnm>PA</fnm></au>
    <au><snm>Besserve</snm><fnm>M</fnm></au>
  </aug>
  <source>Under review, \href{http://xxx.lanl.gov/abs/1202.4482}</source>
  <pubdate>2012</pubdate>
</bibl>

<bibl id="B14">
  <title><p>A neuromorphic architecture for object recognition and motion
  anticipation using burst-{S}{T}{D}{P}</p></title>
  <aug>
    <au><snm>Nere</snm><fnm>A</fnm></au>
    <au><snm>Olcese</snm><fnm>U</fnm></au>
    <au><snm>Balduzzi</snm><fnm>D</fnm></au>
    <au><snm>Tononi</snm><fnm>G</fnm></au>
  </aug>
  <source>PLoS One</source>
  <pubdate>2012</pubdate>
  <volume>7</volume>
  <issue>5</issue>
  <fpage>e36958</fpage>
</bibl>

<bibl id="B15">
  <title><p>Spike timing-dependent plasticity as dynamic filter</p></title>
  <aug>
    <au><snm>Schmiedt</snm><fnm>J</fnm></au>
    <au><snm>Albers</snm><fnm>C</fnm></au>
    <au><snm>Pawelzik</snm><fnm>K</fnm></au>
  </aug>
  <source>Advances in Neural Information Processing Systems (NIPS)</source>
  <pubdate>2010</pubdate>
</bibl>

<bibl id="B16">
  <title><p>Neural {N}etwork {L}earning: {T}heoretical
  {F}oundations</p></title>
  <aug>
    <au><snm>Anthony</snm><fnm>M</fnm></au>
    <au><snm>Bartlett</snm><fnm>PL</fnm></au>
  </aug>
  <publisher>Cambridge Univ Press</publisher>
  <pubdate>1999</pubdate>
</bibl>

<bibl id="B17">
  <title><p>Large {M}argin {C}lassification {U}sing the {P}erceptron
  {A}lgorithm</p></title>
  <aug>
    <au><snm>Freund</snm><fnm>Y</fnm></au>
    <au><snm>Schapire</snm><fnm>R E</fnm></au>
  </aug>
  <source>Machine Learning</source>
  <pubdate>1999</pubdate>
  <volume>37</volume>
  <issue>3</issue>
  <fpage>277</fpage>
  <lpage>296</lpage>
</bibl>

<bibl id="B18">
  <title><p>Decorrelated neuronal firing in cortical microcircuits</p></title>
  <aug>
    <au><snm>Ecker</snm><fnm>AS</fnm></au>
    <au><snm>Berens</snm><fnm>P</fnm></au>
    <au><snm>Keliris</snm><fnm>GA</fnm></au>
    <au><snm>Bethge</snm><fnm>M</fnm></au>
    <au><snm>Logothetis</snm><fnm>NK</fnm></au>
    <au><snm>Tolias</snm><fnm>AS</fnm></au>
  </aug>
  <source>Science</source>
  <pubdate>2010</pubdate>
  <volume>327</volume>
  <issue>5965</issue>
  <fpage>584</fpage>
  <lpage>7</lpage>
</bibl>

<bibl id="B19">
  <title><p>Spike timing-dependent plasticity of neural circuits.</p></title>
  <aug>
    <au><snm>Dan</snm><fnm>Y</fnm></au>
    <au><snm>Poo</snm><fnm>MM</fnm></au>
  </aug>
  <source>Neuron</source>
  <publisher>Division of Neurobiology, Department of Molecular and Cell Biology
  and Helen Wills Neuroscience Institute, University of California, Berkeley,
  Berkeley, CA 94720 USA. ydan@berkeley.edu</publisher>
  <pubdate>2004</pubdate>
  <volume>44</volume>
  <issue>1</issue>
  <fpage>23</fpage>
  <lpage>-30</lpage>
</bibl>

<bibl id="B20">
  <title><p>Time structure of the activity in neural network models</p></title>
  <aug>
    <au><snm>Gerstner</snm><fnm>W</fnm></au>
  </aug>
  <source>Phys. Rev. E</source>
  <pubdate>1995</pubdate>
  <volume>51</volume>
  <issue>1</issue>
  <fpage>738</fpage>
  <lpage>-758</lpage>
</bibl>

<bibl id="B21">
  <title><p>Neural {N}etworks and the {B}ias/{V}ariance {D}ilemma</p></title>
  <aug>
    <au><snm>Geman</snm><fnm>S</fnm></au>
    <au><snm>Bienenstock</snm><fnm>E</fnm></au>
    <au><snm>Doursat</snm><fnm>R</fnm></au>
  </aug>
  <source>Neural Comp</source>
  <pubdate>1992</pubdate>
  <volume>4</volume>
  <fpage>1</fpage>
  <lpage>-58</lpage>
</bibl>

<bibl id="B22">
  <title><p>Experiments with a {N}ew {B}oosting {A}lgorithm</p></title>
  <aug>
    <au><snm>Freund</snm><fnm>Y</fnm></au>
    <au><snm>Schapire</snm><fnm>R E</fnm></au>
  </aug>
  <source>Machine Learning: Proceedings of the Thirteenth International
  Conference</source>
  <pubdate>1996</pubdate>
</bibl>

<bibl id="B23">
  <title><p>Boosting the {M}argin: {A} {N}ew {E}xplanation for the
  {E}ffectiveness of {V}oting {M}ethods</p></title>
  <aug>
    <au><snm>Schapire</snm><fnm>R E</fnm></au>
    <au><snm>Freund</snm><fnm>Y</fnm></au>
    <au><snm>Bartlett</snm><fnm>P</fnm></au>
    <au><snm>Lee</snm><fnm>W S</fnm></au>
  </aug>
  <source>The Annals of Statistics</source>
  <pubdate>1998</pubdate>
  <volume>26</volume>
  <issue>5</issue>
</bibl>

<bibl id="B24">
  <title><p>Theory of classification: {A} survey of some recent
  advances</p></title>
  <aug>
    <au><snm>Boucheron</snm><fnm>S</fnm></au>
    <au><snm>Bousquet</snm><fnm>O</fnm></au>
    <au><snm>Lugosi</snm><fnm>G</fnm></au>
  </aug>
  <source>ESAIM: PS</source>
  <pubdate>2005</pubdate>
  <volume>9</volume>
  <fpage>323</fpage>
  <lpage>375</lpage>
</bibl>

<bibl id="B25">
  <title><p>Metabolic cost as a unifying principle governing neuronal
  biophysics</p></title>
  <aug>
    <au><snm>Hasenstaub</snm><fnm>A</fnm></au>
    <au><snm>Otte</snm><fnm>S</fnm></au>
    <au><snm>Callaway</snm><fnm>E</fnm></au>
    <au><snm>Sejnowski</snm><fnm>TJ</fnm></au>
  </aug>
  <source>Proc Natl Acad Sci U S A</source>
  <pubdate>2010</pubdate>
  <volume>107</volume>
  <issue>27</issue>
  <fpage>12329</fpage>
  <lpage>34</lpage>
</bibl>

<bibl id="B26">
  <title><p>Cascade {M}odels of {S}ynaptically {S}tored {M}emories</p></title>
  <aug>
    <au><snm>Fusi</snm><fnm>S</fnm></au>
    <au><snm>Drew</snm><fnm>PJ</fnm></au>
    <au><snm>Abbott</snm><fnm>LF</fnm></au>
  </aug>
  <source>Neuron</source>
  <pubdate>2005</pubdate>
  <volume>45</volume>
  <fpage>599</fpage>
  <lpage>611</lpage>
</bibl>

<bibl id="B27">
  <title><p>Limits on the memory storage capacity of bounded
  synapses</p></title>
  <aug>
    <au><snm>Fusi</snm><fnm>S</fnm></au>
    <au><snm>Abbott</snm><fnm>LF</fnm></au>
  </aug>
  <source>Nature Neuroscience</source>
  <pubdate>2007</pubdate>
  <volume>10</volume>
  <issue>4</issue>
  <fpage>485</fpage>
  <lpage>-493</lpage>
</bibl>

<bibl id="B28">
  <title><p>Sleep function and synaptic homeostasis.</p></title>
  <aug>
    <au><snm>Tononi</snm><fnm>G</fnm></au>
    <au><snm>Cirelli</snm><fnm>C</fnm></au>
  </aug>
  <source>Sleep Med Rev</source>
  <publisher>Department of Psychiatry, University of Wisconsin, 6001 Research
  Park Blvd., Madison, WI 53719, USA. gtononi@wisc.edu</publisher>
  <pubdate>2006</pubdate>
  <volume>10</volume>
  <issue>1</issue>
  <fpage>49</fpage>
  <lpage>-62</lpage>
</bibl>

<bibl id="B29">
  <title><p>Molecular and electrophysiological evidence for net synaptic
  potentiation in wake and depression in sleep</p></title>
  <aug>
    <au><snm>Vyazovskiy</snm><fnm>V V</fnm></au>
    <au><snm>Cirelli</snm><fnm>C</fnm></au>
    <au><snm>Pfister Genskow</snm><fnm>M</fnm></au>
    <au><snm>Faraguna</snm><fnm>U</fnm></au>
    <au><snm>Tononi</snm><fnm>G</fnm></au>
  </aug>
  <source>Nat Neurosci</source>
  <pubdate>2008</pubdate>
  <volume>11</volume>
  <issue>2</issue>
  <fpage>200</fpage>
  <lpage>8</lpage>
</bibl>

<bibl id="B30">
  <title><p>Cortical firing and sleep homeostasis</p></title>
  <aug>
    <au><snm>Vyazovskiy</snm><fnm>V V</fnm></au>
    <au><snm>Olcese</snm><fnm>U</fnm></au>
    <au><snm>Lazimy</snm><fnm>YM</fnm></au>
    <au><snm>Faraguna</snm><fnm>U</fnm></au>
    <au><snm>Esser</snm><fnm>S K</fnm></au>
    <au><snm>Williams</snm><fnm>J C</fnm></au>
    <au><snm>Cirelli</snm><fnm>C</fnm></au>
    <au><snm>Tononi</snm><fnm>G</fnm></au>
  </aug>
  <source>Neuron</source>
  <pubdate>2009</pubdate>
  <volume>63</volume>
  <issue>6</issue>
  <fpage>865</fpage>
  <lpage>78</lpage>
</bibl>

<bibl id="B31">
  <title><p>Sleep and waking modulate spine turnover in the adolescent mouse
  cortex.</p></title>
  <aug>
    <au><snm>Maret</snm><fnm>S</fnm></au>
    <au><snm>Faraguna</snm><fnm>U</fnm></au>
    <au><snm>Nelson</snm><fnm>AB</fnm></au>
    <au><snm>Cirelli</snm><fnm>C</fnm></au>
    <au><snm>Tononi</snm><fnm>G</fnm></au>
  </aug>
  <source>Nat Neurosci</source>
  <publisher>Department of Psychiatry, University of Wisconsin, Madison,
  Wisconsin, USA.</publisher>
  <pubdate>2011</pubdate>
  <volume>14</volume>
  <issue>11</issue>
  <fpage>1418</fpage>
  <lpage>-1420</lpage>
</bibl>

<bibl id="B32">
  <title><p>Unsupervised learning of visual features through spike timing
  dependent plasticity</p></title>
  <aug>
    <au><snm>Masquelier</snm><fnm>T</fnm></au>
    <au><snm>Thorpe</snm><fnm>SJ</fnm></au>
  </aug>
  <source>PLoS Comput Biol</source>
  <pubdate>2007</pubdate>
  <volume>3</volume>
  <issue>2</issue>
  <fpage>e31</fpage>
</bibl>

<bibl id="B33">
  <title><p>Attention-gated reinforcement learning of internal representations
  for classification.</p></title>
  <aug>
    <au><snm>Roelfsema</snm><fnm>PR</fnm></au>
    <au><snm>Ooyen</snm><fnm>A</fnm></au>
  </aug>
  <source>Neural Comput</source>
  <publisher>Netherlands Ophthalmic Research Institute, 1105 BA Amsterdam, The
  Netherlands. p.roelfsema@ioi.knaw.nl</publisher>
  <pubdate>2005</pubdate>
  <volume>17</volume>
  <issue>10</issue>
  <fpage>2176</fpage>
  <lpage>-2214</lpage>
</bibl>

</refgrp>
} 
}

\newcounter{si-sec}
\addtocounter{si-sec}{1}
\renewcommand{\thesection}{A.\arabic{si-sec}}
\renewcommand{\theequation}{A.\arabic{equation}}

\pagebreak
{\Large{\textbf{Appendices}}}

\section{The perceptron}
\label{s:perceptron}

We describe the perceptron to facilitate comparison with the selectron. 

The perceptron's loss function and learning rule are most naturally expressed when inputs and outputs take values in $\{\pm1\}$. We therefore present the perceptron in both $\pm1$ and $0/1$ ``coordinate systems''.

Let us first relabel inputs and outputs from $0/1$ to $\pm1$: $\Xn=2\X-1$ and $\Yn=2Y-1$. Given input $\xn$, the perceptron's output is determined according to 
\begin{equation*}
	\yn=\rho_\wt(\xn):=\text{sign} \left(\wt^\intercal \xn\right),
\end{equation*}
where  $\wt$ is a real-valued $N$-vector specifying the perceptron's synaptic weights.

Given supervisor $\sigma:\Xn\rightarrow\{\pm1\}$ that labels inputs as belonging to one of two classes, define $0/1$-valued loss function for the perceptron as
\begin{equation*}
	l_p(\xn,\rho_\wt,\sigma) := 
	\bO_{-\rho_\wt\cdot \sigma(\xn)} = 
	\begin{cases}
		0 & \text{if }\rho_\wt(\xn)= \sigma(\xn)\\
		1 & \text{else}.
	\end{cases}
\end{equation*}
The following learning rule 
\begin{equation*}
	\Delta \wt_j =\alpha\cdot \xn_j\cdot\Big(\sigma(\xn)-\rho_\wt(\xn)\Big)
	= \alpha\cdot \begin{cases}
		\xn_{j} & \text{if }\sigma(\xn)=1, \rho_\wt(\xn)=-1\\
		-\xn_{j} & \text{if }\sigma(\xn)=-1, \rho_\wt(\xn)=1\\
		0 & \text{else}
	\end{cases}
\end{equation*}
converges onto an optimal solution if the classes are linearly separable. 

\paragraph{The perceptron in ``$0/1$ coordinates''.}
For the sake of comparison, we reformulate the perceptron in ``$0/1$ coordinates''. If $\wt_j\geq0$ for all $j$, the mechanism of the perceptron is
\begin{equation}
	\label{e:perceptron}
	y=f_\wt(\x):= H\left(\wt^\intercal\x-\frac{\|\wt\|_1}{2}\right).
\end{equation}
Similarly, we obtain loss function
\begin{equation*}
	l_p(\x,f_\wt,\sigma) := 
	\bO_{-(2f_\wt(\x)-1)\cdot \sigma(\x)} = 
	\begin{cases}
		0 & \text{if }\x=1, \sigma(\x)=1 \text{ or }\x=0, \sigma(\x)=-1\\
		1 & \text{else}
	\end{cases}
\end{equation*}
and learning rule
\begin{equation}
	\label{e:prule}
	\Delta \wt_j =\alpha\cdot (2\x_j-1)\cdot\Big(\sigma(\x)-2f_\wt(\x)+1\Big)
	= \alpha\cdot \begin{cases}
		\xn_{j} & \text{if }\sigma(\xn)=1, \rho_\wt(\xn)=-1\\
		-\xn_{j} & \text{if }\sigma(\xn)=-1, \rho_\wt(\xn)=1\\
		0 & \text{else}.
	\end{cases}
\end{equation}

\paragraph{Non-biological features of the perceptron.}
We highlight two features of the perceptron. First, learning rule \eqref{e:prule} is not selective. If the perceptron classifies an input incorrectly, it updates its synaptic weights regardless of whether it outputted a 0 or a 1. It is thus not an output spike-dependent learning rule. The main consequence of this is that the perceptron is forced to classify every input. The selectron, by contrast, actively ignores neuromodulatory signals when it does not spike.

Second, the perceptron requires a local error signal. Multilayer perceptrons are constructed by replacing the $\text{sign}(\bullet)$ in \eqref{e:perceptron} with a differentiable function, such as the sigmoid, and backpropagating errors. However, backpropagation requires two pathways: one for feedforward spikes and another for feedback errors. In other words, the perceptron requires local error signals. A dedicated error pathway is biologically implausible  \cite{roelfsema:05}, suggesting that the cortex relies on alternate mechanisms. 

\addtocounter{si-sec}{1}
\section{Proof of Theorem~\ref{t:frequency}}

The theorem requires computing the second moment of a selectron's total current, given i.i.d. Bernoulli inputs. We also compute the expectation and the variance since these are of intrinsic interest.

\begin{lem}[moments for i.i.d. Bernoulli inputs]\eod
	\label{t:expvar}
	For Bernoulli i.i.d. inputs on synapses, i.e. $P(\x_j=0)=p$ for all $j$, we have
	\begin{align*}
		\bE[\langle \wt,\x\rangle-\vartheta] & = p\cdot\|\wt\|_1-\vartheta \\
		\bV[\langle \wt,\x\rangle-\vartheta] & = p(1-p)\cdot \|\wt\|_2^2\\
		\bE[\langle \wt,\x\rangle^2] & = p(1-p)\cdot \|\wt\|_2^2 + p^2\cdot\|\wt\|_1^2.
	\end{align*}
\end{lem}

\begin{proof}
	For the mean,
	\begin{align*}
		\bE\Big[\langle \wt,\x\rangle\Big] & = \sum_{\x\in\X}\left[P(\x_1)\cdots P(\x_n)\cdot \left(\sum_{j=1}^n \wt_j\cdot \x_j\right)\right]\\
		& = \sum_{j=1}^n P(\x_j=1)\cdot \wt_j
		= \sum_{j=1}^n p\cdot \wt_j
		= p\cdot\|\wt\|_1,
	\end{align*}
	since $\wt_j\geq0$ for all $j$.
	
	For the variance,
	\begin{align*}
		\bV\Big[\langle\wt,\x\rangle\Big] 
		& = \sum_{\x\in\X}P(\x)\langle\wt,\x\rangle^2 - p^2\cdot\|\wt\|_1^2\\
		& =\left[\sum_{i\neq j} p^2\cdot\wt_i\wt_j + \sum_{j} p\cdot\wt_j^2\right]
		-\left[\sum_j p^2\cdot\wt_j^2 + \sum_{i\neq j} p^2\cdot \wt_i\wt_j \right]\\
		& = p(1-p)\cdot \|\wt\|_2^2.
	\end{align*}
	The expression for $\bE[\langle \wt,\x\rangle^2]$ follows immediately.
\end{proof}

\begin{thmtag}{\ref{t:frequency}.}
	\emph{Assuming the inputs on each synapse are i.i.d. Bernoulli variables with $P(\text{spike})=p$, we have
		\begin{equation*}
			P\Big(f_\wt(\x)=1\Big) \leq p\cdot\left(\frac{\|\wt\|_1}{\vartheta}\right)^2. 
		\end{equation*}}
\end{thmtag}

\begin{proof}
	By Lemma~\ref{t:expvar}, $\bE\langle \wt, \x\rangle^2=p(1-p)\cdot \|\wt\|_2^2 + p^2\cdot\|\wt\|_1^2$. Applying Chebyshev's equality, $P(|X|>\epsilon)\leq \frac{\bE X^2}{\epsilon^2}$, obtains 
	\begin{equation*}
		P\Big(f_\wt(\x)=1\Big)= P\Big(\langle \wt, \x\rangle>\vartheta \Big)
		\leq\frac{p(1-p)\cdot \|\wt\|_2^2 + p^2\cdot\|\wt\|_1^2}{\vartheta^2}. 
	\end{equation*}
	The result follows since $\|\wt\|_2\leq \|\wt\|_1$.
\end{proof}

\addtocounter{si-sec}{1}
\section{Proof of Theorem~\ref{t:limit}}
\label{s:stdp-grad}

We first compute the limit for STDP, and then move on to leaky integrate-and-fire neurons.

\begin{lem}[fast time constant limit of STDP]\label{t:fasttime}\eod
	The fast time constant limit of STDP, recall \eqref{e:stdp},  is 
	\begin{equation*}
		\lim_{\tau_\bullet\rightarrow 0}
		\Delta \mathbf{w}_{jk} = (\alpha_+-\alpha_-)\cdot \delta_0(\Delta t)
	\end{equation*}	
	where $\delta_0$ is the Dirac delta.
\end{lem}

\begin{proof}
	We start from the STDP curve:
	\begin{align*}
		\Delta \mathbf{w}_{jk} 
		& := \frac{\alpha_+}{\tau_+} \exp\left(\frac{t_j-t_k}{\tau_+}\right) H(t_k-t_j) 
		- \frac{\alpha_-}{\tau_-} \exp\left(\frac{t_k-t_j}{\tau_-}\right) H(t_j-t_k)\\
		& = \frac{\alpha_+}{\tau_+} \exp\left(\frac{-\Delta t}{\tau_+}\right) H(\Delta t) 
		- \frac{\alpha_-}{\tau_-} \exp\left(\frac{\Delta t}{\tau_-}\right) H(-\Delta t)
	\end{align*}
	where $H$ is the Heaviside function. Note that we divide by $\tau_{\pm}$ to ensure that the area is constant as $\tau_\pm\rightarrow 0$. An equivalent approach is to rescale $\alpha$ and $\tau$ proportionately, so that both tend to zero. 
	
	Let us now compute the limit, in the sense of distributions, as $\tau_\pm\rightarrow 0$. Let $f$ be a test function in the Schwartz space. The linear form associated to
	\begin{equation*}
		S_+(\Delta t) = \frac{\alpha_+}{\tau_+} \exp\left(\frac{-\Delta t}{\tau_+}\right) H(\Delta t) 
	\end{equation*}
	is 
	\begin{equation*}
		{\mathcal S}_+:f(x)\mapsto \int S_+(x)f(x)dx.
	\end{equation*}
	Thus, 
	\begin{equation*}
		{\mathcal S}_+(f) = \int \frac{\alpha_+}{\tau_+} \exp\left(\frac{-x}{\tau_+}\right) H(x) f(x)dx.
	\end{equation*}
	Integrating by parts obtains 
	\begin{equation*}
		{\mathcal S}_+(f) =\left[-\alpha_+ \exp\left(\frac{-x}{\tau_+}\right)f(x)\right]^\infty_0
		- \int_0^\infty -\alpha_+ \exp\left(\frac{-x}{\tau_+}\right)f'(x)dx
	\end{equation*}
	so that
	\begin{equation*}
		{\mathcal S}_+(f) = \alpha_+\cdot f(0) +\tau_+\cdot {\mathcal S}_+(f')
	\end{equation*}
	and $\lim_{\tau_+\rightarrow 0}\left[{\mathcal S}_+(f)\right]=\alpha_+\cdot f(0)=\alpha_+\cdot \delta_0(f)$. A similar argument holds for the negative part of the STDP curve, so the fast time constant limit is 
	\begin{equation*}
		\Delta \mathbf{w}_{jk} = (\alpha_+-\alpha_-)\cdot \delta_0(\Delta t).
	\end{equation*}
\end{proof}

\begin{rem}
	If STDP incorporates a potentiation bias (in the sense that the net integral is positive \cite{song:00}), then the fast time constant limit acts exclusively by potentiation.
\end{rem}

\begin{thmtag}{\ref{t:limit}.}{}
\emph{In the fast time constant limit, $\lim_{\tau_\bullet}\rightarrow 0$, the SRM transforms into a selectron with
	\begin{equation*}
		f_\wt(t) = H\Big(M_\wt(t)-\vartheta\Big)
		\quad\text{ where }\quad M_\wt = \sum_{\{j|t_j\geq t_k\}}\wt_{jk}\cdot \delta_{t_k}(t).
	\end{equation*}
	Moreover, STDP transforms into learning rule \eqref{e:grad-asc} in the unsupervised setting with $\nu(\x)=1$ for all $\x$. Finally, STDP arises as gradient ascent on a reward function whose limit is the unsupervised setting of reward function $\eqref{e:reward}$.
}
\end{thmtag}

\begin{proof}
	Setting $K_1=K_2=0$, $\tau_s=\frac{1}{2}\tau_m$, $\frac{1}{K}=e^{-1} - e^{-2}$, and taking the limit $\tau_m\rightarrow 0$ yields
	\begin{equation*}
		M_{\wt}(t) = \sum_{\{j|t_j\geq t_k\}}\wt_{jk}\cdot \delta_{t_k}(t).
	\end{equation*}
	By Lemma~\ref{t:fasttime}, taking limits $\tau_\pm\rightarrow0$ transforms STDP into 
	\begin{equation*}
		\Delta \wt_{jk} =\begin{cases}
	 	(\alpha_+ -\alpha_-) & \mbox{ if }t_k=t_j
		\\  
		0 & \mbox{ else},
		\end{cases}		
	\end{equation*}
	which is a special case of the online learning rule for the selectron derived in \S\ref{s:selectron}, where the neuromodulatory response is $\nu(\x)=1$ for all $\x$.
	
	STDP is easily seen to arise as gradient ascent on
	\begin{equation}
		\label{e:stdp-opt}
		\arg\max_{\mathbf{w}}\sum_{t_k}\left[\sum_{t_{k-1}< t_j \leq t_k} \wt_j \cdot e^{\left(\frac{t_j-t_k}{\tau_+}\right)}
		-\sum_{t_k<t_j <t_{k+1}}\wt_j \cdot e^{\left(\frac{t_k-t_j}{\tau_-}\right)}\right].
	\end{equation}
	Taking the limit of \eqref{e:stdp-opt} yields the special case of \eqref{e:reward} where $\nu(\x)=1$ for all $\x$.
\end{proof}

Eq.~\eqref{e:stdp-opt} can be expressed in the shorthand
\begin{equation*}
	\arg\max_{\mathbf{w}}\left[
	\sum_{t_k}\big(\mathbf{w}^\intercal\cdot \mathbf{d}(t_k)\big)\cdot f_{\mathbf{w}}(t_k)\right],	
	\text{ where }	\mathbf{d}_j(t) = \begin{cases}
e^{\left(\frac{t_j-t}{\tau_+}\right)} & \mbox{if }t_j\leq t
\\  
-e^{\left(\frac{t-t_j}{\tau_-}\right)} & \mbox{else}
\end{cases}.
\end{equation*}

\addtocounter{si-sec}{1}
\section{Proof of Theorem~\ref{t:lbound}}

To prove the theorem, we first recall an error bound from \cite{boucheron:05}. To state their result, we need the following notation. Let 
\begin{equation*}
	{\mathcal C}_{\omega} = \left\{f(\x)=\left.\text{sign}\left(\sum_{j=1}^N{\mathbf a}_j\cdot g^j(\x)\right) \;\right|\; 
	{\mathbf a}_j \in\bR, \|{\mathbf a}\|_1\leq \omega
	\text{ and } g^j\in {\mathcal C}\right\}
\end{equation*}
where ${\mathcal C}$ is a class of base classifiers taking values in $\pm1$.

Let $\phi:\bR\rightarrow \bR_+$ be a nonnegative cost function such that $\bO_{x>0}\leq \phi(x)$, for example $\phi(x)=(1-x)_+$. Define 
\begin{eqnarray}
		\label{e:L}
	L(f)=\bE\bO_{-f(X)Y<0} & \text{and} & \widehat{L}_n(f)=\frac{1}{n}\sum_{i=1}^n \bO_{-f(X_i)Y_i<0}
\end{eqnarray}
and
\begin{eqnarray}
		\label{e:A}
	A(f)=\bE\phi(-f(X)Y) & \text{and} & \widehat{A}_n(f)=\frac{1}{n}\sum_{i=1}^n \phi(-f(X_i)Y_i).
\end{eqnarray}

\begin{thm}\label{t:boucheron}
	Let $f$ be a function chosen from ${\mathcal C}_{\omega}$ based on data $(X_i,Y_i)_{i=1}^n$. With probability at least $1-\delta$,
	\begin{equation}
		L(f)\leq \widehat{A}(f_n)
		+ 2L_\phi\cdot \bE{\mathcal Rad}_n({\mathcal C}(X_1^n|f=1))
		+ B\sqrt{\frac{2\log\frac{1}{\delta}}{n}},
	\end{equation}
	where $L_\phi$ is the Lipschitz constant of $\phi$, $B$ is a uniform upper bound on $\phi(-f(x)y)$, and ${\mathcal Rad}\big({\mathcal C}(X_1^n)\big)$ is the Rademacher complexity of ${\mathcal C}$ on data $X_1^n$.
\end{thm}

\begin{proof}
	See \S4.1 of \cite{boucheron:05}.
\end{proof}

We are interested in errors of commission, not omission, so we consider modified versions of \eqref{e:L},
\begin{eqnarray}
	\label{e:mL}
	\cL(f)=\bE\left[\bO_{-f(X)Y<0}\cdot\bO_{f>0}\right]
	& \text{and} &
	\widehat{\cL}_n(f)=\frac{1}{n}\sum_{i=1}^n\left[\bO_{-f(X)Y<0}\cdot\bO_{f>0}\right],
\end{eqnarray}
and \eqref{e:A},
\begin{eqnarray}
	\label{e:mA}
	\cA(f)=\bE\left[\phi(-f(X)Y)\cdot\bO_{f>0}\right] & \text{and} 
	& \widehat{\cA}_n(f)=\frac{1}{n}\sum_{i=1}^n \left[\phi(-f(X_i)Y_i)\cdot\bO_{f>0}\right].
\end{eqnarray}
where we multiply by the indicator function $\bO_{f>0}$. This results in a discontinuous hinge function. We therefore have to modify the above theorem.

\begin{thm}\label{t:boucheron-adapt}
	Let $f$ be a function chosen from ${\mathcal C}_{\omega}$  based on data $(X_i,Y_i)_{i=1}^n$. With probability at least $1-\delta$,
	\begin{equation*}
		\cL(f) \leq \widehat{\cA}_n(f) 
		+ 2(\bE[\bO_f]\cdot L_\phi+ B)\cdot \bE{\mathcal Rad}_n({\mathcal C}(X_1^n|f=1))
		+ 2B\sqrt{\frac{2\log\frac{2}{\delta}}{n}},
	\end{equation*}
	where $L_\phi$ is the Lipschitz constant of $\phi$, $B$ is a uniform upper bound on $\phi(-f(x)y)$, and ${\mathcal Rad}\big({\mathcal C}(X_1^n)\big)$ is the Rademacher complexity of ${\mathcal C}$ on data $X_1^n$.
\end{thm}

\begin{proof}
	We adapt the proof in \cite{boucheron:05}. By definition it follows that
	\begin{align*}
		\cL(f) & 
		\leq \cA(f) \\
		& = \widehat{\cA}_n(f) + \cA(f) - \widehat{\cA}_n(f) 
	\end{align*}
	Now observe that $\sum_x \left[p(x)f(x)\bO_S\right]=\sum_{x\in S}p(x)f(x)=p(S)\sum_x q(x)f(x)$, where $q(x):= p(x|x\in A)$. Thus, changing distributions from $P(x)$ to $Q(x)=P(x|f=1)$, we can write
	\begin{align*}
		\cL(f) & \leq \widehat{\cA}_n(f) + 
		\left[P(f=1)\cdot\bE_{Q}[\phi(-f(X)Y)]-\widehat{P}_n(f=1)\cdot\sum_{\{i:f(x_i=1\}} \frac{\phi(-f(x_i)y_i)}{|\{i:f(x_i=1\}|}\right]\\
		& = \widehat{\cA}_n(f) + P(f=1)\cdot\left[A_Q(f)-\widehat{A}_{\widehat{Q}_n}(f)\right]
		+ \widehat{A}_{\widehat{Q}_n}(f)\cdot\left[\bE_P f-\bE_{\widehat{P}_n}f\right]
	\end{align*}
	where $\widehat{P}_n(x_i)=\frac{1}{n}$ is the empirical distribution, $\widehat{P}_n(f=1)=\frac{|\{i:f(x_i=1\}|}{n}$ and $\widehat{Q}_n(x_i)=\frac{1}{|\{i:f(x_i=1\}|}$.

	Continuing,
	\begin{align*}
		\cL(f) & \leq \widehat{\cA}_n(f) + \bE_P[\bO_f]\cdot \sup_{g\in{\mathcal C}}\left[A_Q(g)-\widehat{A}_{\widehat{Q}_n}(g)\right] + 
		\widehat{A}_{\widehat{Q}_n}(f)\cdot\sup_{g\in{\mathcal C}}\left[\bE_P g -\bE_{\widehat{P}_n}g\right]\\
		& \leq \widehat{\cA}_n(f) + 2\bE[\bO_f]\cdot L_\phi\cdot \bE{\mathcal Rad}_n({\mathcal C}(X_1^n|f=1))\\
		& \quad\quad+ 2\widehat{A}_{\widehat{Q}_n}\cdot \bE{\mathcal Rad}_n({\mathcal C}(X_1^n|f=1))+
		2B\sqrt{\frac{2\log\frac{2}{\delta}}{n}}\\
		& \leq \widehat{\cA}_n(f) 
		+ 2(\bE[\bO_f]\cdot L_\phi+ B)\cdot \bE{\mathcal Rad}_n({\mathcal C}(X_1^n|f=1))
		+ 2B\sqrt{\frac{2\log\frac{2}{\delta}}{n}}
	\end{align*}
	where we bound the two supremum's with high probability separately using Rademacher complexity and apply the union bound. The last inequality follows since $\phi\leq B$.
\end{proof}
\begin{rem}
	The change of distribution is not important, since the Rademacher complexity is bounded by the distribution free VC-dimension in Lemma~\ref{t:radem}.
\end{rem}

To recover the setting of Theorem~\ref{t:lbound}, we specialize to $\phi^\kappa(f)=(\kappa-f(X)Y)_+$, for $\kappa\geq 1$, so that the hinge loss can be written as $h^\kappa = \phi^\kappa\big(\wt^\intercal \x-\vartheta\big)\cdot \bO_{f(X)>0}$. The Lipschitz constant is $L_\phi=1$.

Now let 
\begin{equation*}
	\cF_{\omega} = \left\{f(\x)=\left.H\left(\sum_{j=1}^N\wt^j\cdot g^j(\x)-\vartheta\right) \;\right|\; 
	\|\wt\|_1\leq \omega
	\text{ and } g^j\in \cF\right\},
\end{equation*}
where functions in $\cF$ take values in $0/1$.

Function class $\cF_{\omega}$ denotes  functions implementable by a two-layer network of selectrons $S^k=\{n^k\}\cup\{n^j:n^j\rightarrow n^k\}$. The outputs of $g^j$ are aggregated, so that the function class $\cF_{\omega}$ of subnetwork $S^k$ is larger than that of selectrons considered individually, i.e. $\cF$.

\begin{lem}\label{t:radem}
	The Rademacher complexity of $\cF_{\omega}$ is upper bounded by
	\begin{equation*}
		{\mathcal Rad}_n(\cF_{\omega})
		\leq \omega\cdot\frac{\sqrt{2(N+1)\log(n+1)}+\frac{1}{2}}{\sqrt{n}}.
	\end{equation*}
\end{lem}

\begin{proof}
	Given selectron $f_{\wt}(\x)=\sum_j \wt_j\cdot g^j_{\wt}(\x)$, let $\bar{g}^j:=2g^j-1$ be the corresponding $\{\pm1\}$-valued function. Then
\begin{equation*}
	\sum_{j=1}^N \wt_j g^j(\x) -\vartheta = \frac{1}{2}\sum_{j=1}^N \wt_j (\bar{g}^j(\x)+1)-\vartheta
	=\frac{1}{2}\|\wt\|_1-\vartheta+\frac{1}{2}\sum_{j=1}^N \wt_j \bar{g}^j(\x),
\end{equation*}
since $\wt_i\geq0$ for all $i$. Thus,
\begin{equation*}
	{\mathcal Rad}_n(\cF_{\omega})
	\leq \frac{\left|\frac{1}{2}\omega-\vartheta\right|}{\sqrt{n}}+ \omega\cdot {\mathcal Rad}_n(\cF).
\end{equation*}
The Rademacher complexity of $\cF$ is upper bounded by the VC-dimension, 
\begin{equation*}
	{\mathcal Rad}({\mathcal F}(X_1^n))\leq \sqrt{\frac{2V_{\mathcal F}\cdot \log(n+1)}{n}}
\end{equation*}
which for a selectron with $N$ synapses is  at most $N+1$.

Finally, note that if $\vartheta<0$ then the selectron always spikes, and if $\vartheta>\|\wt\|_1$ then it never spikes. We therefore assume that $0\leq\vartheta\leq\omega$, which implies $|2\vartheta -\omega|\leq \omega$ and so
\begin{equation*}
	\frac{\left|\frac{1}{2}\omega-\vartheta\right|}{\sqrt{n}} 
	\leq \frac{\omega}{2\sqrt{n}}.
\end{equation*}
\end{proof}

\begin{thmtag}{\ref{t:lbound}.}{}
\emph{Suppose each selectron has $\leq N$ synapses. For any selectron $n^k$, let $S^k=\{n^k\}\cup \{n^j:n^j\rightarrow n^k\}$ denote a 2-layer feedforward subnetwork. 
	For all $\kappa\geq 1$, with probability at least $1-\delta$,
	\begin{align*}
		\bE\underbrace{\big[l(\x,f_\wt,\nu)\big]}_{0/1\text{ loss}}
		\leq & \frac{1}{n}\sum_i \underbrace{h^\kappa\big(\x^{(i)},f_\wt,\nu(\x^{(i)})\big)}_{\text{hinge loss}}
		+ \omega\cdot\underbrace{2B\cdot
		\frac{\sqrt{8(N+1)\log(n+1)}+1}{\sqrt{n}}}_{\text{capacity term}}\\
		& + B\cdot\underbrace{\sqrt{\frac{2\log\frac{2}{\delta}}{n}}}_{\text{confidence term}}
		\hspace{15mm}\text{where }B=\kappa+\omega-\vartheta.
	\end{align*}
}
\end{thmtag}

\begin{proof}
	Applying Theorem~\ref{t:boucheron-adapt}, Lemma~\ref{t:radem}, and noting that $\bE[\bO_f]\leq1\leq B=\kappa+\omega-\vartheta$, where $B$ is an upper bound on the hinge loss, obtains the result.
\end{proof}

\addtocounter{si-sec}{1}
\section{An alternate error bound with hard margins}
\label{s:hmargin}

For the sake of completeness, we prove Corollary~\ref{t:hardmargin}, an alternate to Theorem~\ref{t:lbound} that bounds a symmetric  $0/1$ loss: 
\begin{equation*}
	\bO_{(2f_\wt(\x)-1)\cdot\nu(\x)<0}=\begin{cases}
		1 & \text{if } f_\wt(\x)=1,\nu(\x)=-1\\
		1 & \text{if } f_\wt(\x)=0,\nu(\x)=1\\
		0 & \text{else}.
	\end{cases}
\end{equation*}
The loss penalizes the selectron when either (i) it fires when it shouldn't or (ii) it doesn't fire when it should.

We replace the modified hinge loss in Theorem~\ref{t:lbound} with a hard-margin loss. Following \cite{boucheron:05}, let 
\begin{equation*}
	\widehat{L}_n^\gamma(f) = \frac{1}{n}\sum_{i=1}^{n}\bO_{f(X_i)Y_i<\gamma}
	\;\;\;\text{ and }\;\;\;
	\phi^\gamma(x)=\begin{cases}
		0 & \text{if }x\leq-\gamma\\
		1 & \text{if }x\geq 0\\
		1+x/\gamma & \text{else}.
	\end{cases}
\end{equation*}

The following corollary of Theorem~\ref{t:boucheron} is shown in \cite{boucheron:05}.
\begin{cor}
	\label{t:bhm}
	Let $f_n$ be a function chosen from ${\mathcal C}_{\omega}$. For any $\gamma>0$, with probability at least $1-\delta$,
	\begin{equation*}
		L(f_n) \leq \widehat{L}_n^\gamma(f_n) + \frac{\omega}{\gamma}\sqrt{\frac{2V_{\mathcal C}\log(n+1)}{n}} 
		+ \sqrt{\frac{2\log\frac{1}{\delta}}{n}}.
	\end{equation*}
	where $V_{\mathcal C}$ is the VC-dimension of ${\mathcal C}$.
\end{cor}

We use Corollary~\ref{t:bhm} to derive an alternate error bound for the selectron. Introduce hard-margin loss
\begin{equation*}
	\bO_{\big(\wt^\intercal\x-\vartheta\big)\cdot\nu(\x)<\gamma} = \begin{cases}
		1 & \text{if } \text{sign}(\wt^\intercal \x-\vartheta)=\text{sign}(\nu(\x)) \text{ and }|\wt^\intercal \x-\vartheta|>\gamma\\
		0 & \text{else}
	\end{cases}
\end{equation*}

The following error bound exhibits a trade-off on the $0/1$ loss of a selectron that is controlled by $\gamma$. 

\begin{cor}[error bound with hard margins]\label{t:hardmargin}\eod
	For any $\gamma>0$, with probability at least $1-\delta$,
	\begin{equation*}
		\bE\bO_{(2f_\wt(\x)-1)\cdot\nu(\x)<0} \leq 
		\frac{1}{n}\sum_i \bO_{\big(\wt^\intercal\x^{(i)}-\vartheta\big)\cdot\nu(\x^{(i)})<\gamma} 
		+ \frac{\omega}{\gamma}\cdot\frac{\sqrt{8(N+1)\log(n+1)}+1}{\sqrt{n}}
		+ \sqrt{\frac{2\log\frac{1}{\delta}}{n}}.
	\end{equation*}
\end{cor}

\begin{proof}
	Follows from Lemma~\ref{t:radem} and Corollary~\ref{t:bhm}.
\end{proof}

As $\gamma$ increases, the size of the margin required to avoid counting towards a loss increases. However, the capacity term is multiplied by $\frac{1}{\gamma}$, and so reduces as the size of $\gamma$ increases.

Thus, the larger the margin, on average, the higher the value we can choose for $\gamma$ without incurring a penalty, and the better the bound in Corollary~\ref{t:hardmargin}.

\addtocounter{si-sec}{1}
\section{Proof of Theorem~\ref{t:composition}}
\label{s:composition-proof}

It is helpful to introduce some notation. Given $\x\in\X$, let $S^j_a = \{\x|\x_j=a\}$. Let $\nu_{jk=11}$, defined by equation
\begin{equation*}
	\nu_{jk=11}\cdot \left[\sum_{\x\in  S^j_1} P(\x|\x_j=1) (\wt^\intercal \x-\vartheta)f_\wt(\x)\right]
	= \sum_{\x\in  S^j_1} P(\x|\x_j=1)\nu(\x) (\wt^\intercal \x-\vartheta)f_\wt(\x),
\end{equation*}
quantify the average contribution of neuromodulators when selectrons $n^j$ and $n^k$ both spike. Similarly, let $\nu_{jk=01}$ quantify the average contribution of neuromodulators when $n^k$ spikes and $n^j$ does not, i.e. the sum over $\x\in S^j_0$.

If upstream neurons are reward maximizers then spikes by $n^j$ should predict higher neuromodulatory rewards, on average, than not firing. We therefore assume
\begin{equation}
	\nu_{jk=11}\geq \nu_{j=10}.
	\tag{*}
	\label{e:assumption}
\end{equation}

\begin{thmtag}{\ref{t:composition}}{}
	\emph{
		Let $p_j:=\bE[Y^j]$ denote the frequency of spikes from neuron $n^j$. Under assumption \eqref{e:assumption}, the efficacy of $n^j$'s spikes on $n^k$ is lower bounded by
	\begin{equation*}
		\frac{1}{\nu_{jk=11}}\cdot \frac{\delta R^k}{\delta \x_j}
		\geq \frac{\wt_j\cdot\bE[Y^jY^k]}{p_j}
		+ \frac{2\bE\Big[Y^jY^k\cdot \big((\wt^{\bcancel j})^\intercal \x-\vartheta\big)\Big]}{p_j(1-p_j)}
		- \frac{\bE\Big[Y^k\cdot \big((\wt^{\bcancel j})^\intercal \x-\vartheta\big)\Big]}{1-p_j}
	\end{equation*}	
	where 
	\begin{equation*}
		\wt^{\bcancel j}_i := \begin{cases}
			\wt_i & \text{if }i\neq j\\
			0 & \text{else}.
		\end{cases}
	\end{equation*}
		}
\end{thmtag}

\begin{proof}
	\begin{align*}
		\frac{\delta R^k}{\delta \x_j} & = \bE[R^k|\x_j=1] - \bE[R^k|\x_j=0]\\
		& = \sum_{\x\in  S^j_1} P(\x|\x_j=1)\nu(\x) (\wt^\intercal \x-\vartheta)f_\wt(\x)
		- \sum_{\x\in S^j_0} P(\x|\x_j=0)\nu(\x) (\wt^\intercal \x-\vartheta)f_\wt(\x)\\
		& = \nu_{jk=11}\cdot \sum_{\x\in  S^j_1} P(\x|\x_j=1) (\wt^\intercal \x-\vartheta)f_\wt(\x)
		- \nu_{jk=01}\cdot\sum_{\x\in S^j_0} P(\x|\x_j=0) (\wt^\intercal \x-\vartheta)f_\wt(\x).
	\end{align*}
	Assumption \eqref{e:assumption} implies
	\begin{align*}
		\frac{1}{\nu_{jk=11}}\cdot \frac{\delta R^k}{\delta \x_j} 
		& \geq \sum_{\x\in  S^j_1} P(\x|\x_j=1)(\wt^\intercal \x-\vartheta) f_\wt(\x)
		- \sum_{\x\in S^j_0} P(\x|\x_j=0) (\wt^\intercal \x-\vartheta)f_\wt(\x)\\
		& = \sum_{\x\in  S^j_1} \frac{P(\x)}{P(\x_j=1)} (\wt^\intercal \x-\vartheta)f_\wt(\x)
		- \sum_{\x\in S^j_0} \frac{P(\x)}{P(\x_j=0)} (\wt^\intercal \x-\vartheta)f_\wt(\x)\\
		& = \frac{1}{p_j}\bE\left[Y^jY^k\cdot (\wt^\intercal \x-\vartheta)\right]
		- \frac{1}{1-p_j}\bE\left[(1-Y^j)Y^k \cdot(\wt^\intercal \x-\vartheta)\right]\\
	\end{align*}
	The result follows by direct computation.
\end{proof}

\addtocounter{si-sec}{1}
\section{Generation of Poisson spike trains}
\label{s:poisson}

Spike trains are generated following \cite{Masquelier:2007vn}. Neurons have 200 synaptic inputs. Time is discretized into $1ms$ bins. At each time step spikes are generated according to a Poisson process where the rate varies as follows:
\begin{enumerate}
	\item a synapse has probability $r\cdot dt$ of emitting a spike where $r$ is clipped in $[0,90]Hz$.
	\item $dr=s\cdot dt$ where $s$ is clipped in $[-1800,1800]Hz$.
	\item the rate of change $ds$ of $s$ is uniformly picked from $[-360,360]Hz$.
\end{enumerate}
The resulting spike train has an average firing rate of about $44Hz$. Masquelier \emph{et al}
also add a mechanism to ensure each synapse transmits a spike after at most $50ms$ which we do not implement.

Repeated patterns are sampled from the process above for $50ms$. The pattern is then cut-and-paste over the original Poisson spike train for $\frac{1}{3}$ of the total number of $50ms$ blocks. The cut-and-pasting is imposed on a randomly chosen (but fixed) subset containing $\frac{1}{2}$ of the neuron's synapses.


\end{document}